\documentclass[conference]{IEEEtran}
\IEEEoverridecommandlockouts
\usepackage{cite}
\usepackage{booktabs}
\usepackage{amssymb,amsfonts, mathtools}
\usepackage{amsthm}

\usepackage{algorithm}
\usepackage[noend]{algpseudocode}
\usepackage{graphicx}
\usepackage{textcomp}
\usepackage{xcolor}
\usepackage{subfig}
\usepackage{bm}
\usepackage{indentfirst}
\usepackage{lipsum}
\usepackage{url}

\newtheorem{theorem}{Theorem}
\newtheorem{lemma}[theorem]{Lemma}

\theoremstyle{remark}
\newtheorem{example}{Example}

\def\BibTeX{{\rm B\kern-.05em{\sc i\kern-.025em b}\kern-.08em
    T\kern-.1667em\lower.7ex\hbox{E}\kern-.125emX}}
\begin{document}

\title{Heterogeneity-aware Gradient Coding for Straggler Tolerance}

\author{
	\IEEEauthorblockN{
        Haozhao Wang\IEEEauthorrefmark{1},
        Song Guo\IEEEauthorrefmark{2},
        Bin Tang\IEEEauthorrefmark{3},
        Ruixuan Li\IEEEauthorrefmark{1} and
		Chengjie Li\IEEEauthorrefmark{1}
	}
	\IEEEauthorblockA{
		\IEEEauthorrefmark{1}School of Computer Science and Technology,
		Huazhong University of Science and Technology, Wuhan 430074, China\\
		\IEEEauthorrefmark{2}Department of Computing,
		The Hong Kong Polytechnic University, Hung Hom, Kowloon, Hong Kong\\
		\IEEEauthorrefmark{3}National Key Laboratory for Novel Software Technology,
		Nanjing University, Nanjing 210023, China\\
		Email: \{hz\_wang, rxli, cjl1720\}@hust.edu.cn, cssongguo@comp.polyu.edu.hk, tb@nju.edu.cn
	}
}

\maketitle

\begin{abstract}

Gradient descent algorithms are widely used in machine learning. In order to deal with huge
volume of data, we consider the implementation of gradient descent algorithms in a distributed
computing setting where multiple workers compute the gradient over some partial data and the
master node aggregates their results to obtain the gradient over the whole data. However, its
performance can be severely affected by straggler workers. Recently, some coding-based
approaches are introduced to mitigate the straggler problem, but they are efficient only
when the workers are homogeneous, i.e., having the same computation capabilities. In this
paper, we consider that the workers are heterogenous which are common in modern distributed
systems. We propose a novel heterogeneity-aware gradient coding scheme which can not only
tolerate a predetermined number of stragglers but also fully utilize the computation
capabilities of heterogenous workers. We show that this scheme is optimal when the computation
capabilities of workers are estimated accurately. A variant of this scheme is further proposed
to improve the performance when the estimations of the computation capabilities are not so
accurate. We conduct our schemes for gradient descent based image classification on
QingCloud clusters. Evaluation results show that our schemes can reduce the whole computation
time by up to $3\times$ compared with a state-of-the-art coding scheme.
\end{abstract}


\section{Introduction}
    With the rapid increasing of data size, fast processing of big data becomes more and
more important. Due to the saturation of Moore's law, distributed processing has been viewed
as the primary method for breaking down the limitation of computing power. Modern systems
for distribute processing of big data like MapReduce \cite{DBLP:journals/cacm/DeanG08} and
Apache Spark \cite{DBLP:conf/hotcloud/ZahariaCFSS10} usually adopt a master-slave
architecture. In such architecture, a master server divides the initial task into many
small tasks and assigns them to several slave nodes (worker). These workers process tasks
in parallel and return outcomes back to master after finishing.

    In such distributed form, the performance of distributed system is usually limited by
delays or faults as master collects outcomes from workers \cite{dean2013tail}. Delays or
faults are usually incurred by stragglers which are workers that cannot return outcome
within a reasonable deadline. Stragglers are mainly caused by two reasons, 1) transient
fluctuation of resource in cluster, e.g., fault occurrence \cite{yan2016tr, harlap2017proteus},
resource contention between processes, and 2) consistent heterogeneity of clusters
\cite{jiang2017heterogeneity}. Due to the notable negative impact of stragglers on
performance, many recent works were proposed trying to mitigate them regarding to
different tasks \cite{DBLP:conf/nsdi/AnanthanarayananGSS13, zhang2014dynamic,ananthanarayanan2014grass}.
In this paper, we focus on the task of gradient computing. Gradient is the derivative of objective function and is of great importance
for being the cornerstone of many optimization algorithms \cite{DBLP:conf/nips/CutkoskyB18, DBLP:journals/corr/KingmaB14}.
For gradient computing task Tandon \cite{tandon2017gradient} proposes using coding
method to tolerate stragglers. In their framework, the gradient of a sample is computed
by several workers so that the gradient of the sample could be recovered by master as long as
master receives the update of any worker that participates in the gradient computation of the
sample. The essence of this gradient coding method is to improve stragglers tolerance
by making data duplication. Though their method works efficiently for stragglers incurred
trasient fluctuation, it can do nothing for stragglers caused by heterogeneneity. This is
because it does not take computing capabilities of workers into account as designing coding
scheme. Another work \cite{maity2018robust} encodes the second-moment of data to reduce
computational overhead of encoding naive data. However, it is only limited to the gradient of
linear model which cannot be used in many domains, e.g, training of DNN.

    Considering all the insufficiencies of existing methods, we seek to tolerate stragglers
incurred both the two reasons, i.e., stragglers in heterogeneous clusters, such that
the processing efficiency of distributed system could be improved. This is a non-trival problem,
because heterogeneneity is very common in modern clusters \cite{zhang2014dynamic, zhao2014improving, jiang2017heterogeneity}.
In fact, we can solve this problem by designing a solution that can both tolerate transient
stragglers and take full utilization of the computing resources in heterogeneous cluster.
To acheive this goal, we propose two heterogeneity-aware gradient coding methods
that adaptively allocate data partitions to each worker according to their computing capabilities.
In this way, each worker has the similiar completion time so that the consistent stragglers
incurred by heterogeneity could be eliminated. On the other hand, the transient stragglers
will also be eliminated by using coding theory.

    To implement heterogeneity-aware gradient coding scheme, data partitions
are firstly allocated to each worker according to their processing speed, and then we
show how to construct coding strategy. The experiemental evaluations were done on popular deep
learning tasks on several heterogeneous clusters range from $8$ workers to $48$ workers.
Results show that our methods improve the performance of deep learning task up to
$3\times$ compared to traditional gradient coding methods.

    Our contributions are summarized as follows:
    \begin{itemize}
        \item Straggler tolerance in heterogeneous setup is of great importance, but is
        ignored by existing methods. We propose a new heterogeneity-aware gradient coding
        scheme that could work efficiently in heterogeneous clusters while tolerating stragglers.
        \item We theoretically show that our heterogeneity-aware gradient coding scheme
        is optimal for a cluster with accurately estimated computing capacity.
        \item Considering practicalities of running system that the computing capacity is hard to be
        measured accurately, we further propose a more effcient variant of heterogeneity-aware
        gradient coding scheme.
        \item We conduct our coding schemes for gradient-based machine learning tasks on
        QingCloud clusters. Evaluation results show that our coding scheme could not only tolerate stragglers
        but also take fully utilization of computing capabilities of workers.
    \end{itemize}

    \indent This paper is organized as follows. The related work about stragglers in
distributed system is firstly presented in Section \uppercase\expandafter{\romannumeral2}.
And then, we present the problem formulation in Section \uppercase\expandafter{\romannumeral3}.
After that, we present our designed two heterogeneity-aware gradient coding schemes,
\textit{heter-aware} and \textit{group-based} coding scheme. In Section
\uppercase\expandafter{\romannumeral6}, a wide range of evaluations are performed in
various large-scale heterogeneous clusters to show the efficiency of our coding scheme.
Finally, the conclusions are drawn in Section \uppercase\expandafter{\romannumeral7}.

\section{Related Work}
    Straggler problem has a long history in parallel computing, and it attracts more and more interests
as the era of big data comes. Here below, we will firstly introduce methods for straggler problem from
specific to general, and then show recently emerging coded methods for straggler mitigation.

    \indent Considering distributed learning task is the typical task by using gradient, we firstly
the related work in distributed learning system for stragglers.  Due to fault tolerance property inherent
in machine learning task, there are many methods trying to starting with parallel mechanism.
Typical algorithms including asynchronous parallel training algorithms including
TAP\cite{smola2010architecture, dean2012large} and SSP\cite{ho2013more, cipar2013solving, cui2014exploiting}
were proposed to avoid stragglers in learning steps, where the core idea of these methods is to improve hardware
efficiency by sacrificing statistical efficiency (e.g., convergence accuracy and speed)\cite{hadjis2016omnivore}.
Further based on SSP, DynSSP\cite{jiang2017heterogeneity} was proposed to improve the statistical efficiency of
asynchronous learning by tuning learning rates. Though such parallel algorithms could reduce the affecting
of stragglers, they are hard to analysis, debug, and reproduced. Besides, the accuracy as convergence couldn't
reach optimal as shown in\cite{chen2016revisiting}. Different from these work, we try to mitigate stragglers
for the BSP distribtued scheme that keeps accuracy. \\
    \indent Another line for mitigating stragglers is load balancing which can be referenced by general task.
There are many work \cite{blumofe1999scheduling, acar2013scheduling, dinan2009scalable} trying to rebalance
workload allocation by using work stealing in traditional parallel computing. Work stealing in fact is a
technique that reallocates tasks from busy cores to idle cores. However, this idea isn't suitable to
machine learning task, especially DNN's training. One reason is that each iteration of DNN's training is
very short which lasts only a few seconds or less
\cite{DBLP:journals/corr/ChenLLLWWXXZZ15, DBLP:conf/osdi/2016, goyal2017accurate} causing that the
detection of stragglers and transferring of workloads are almost impossible. In this paper, we propose a
new load balancing method that use the property of data parallel processing task that the
computing complexity is of each task is proportional to its number of samples.

    \indent Recently, coding theory based methods were also introduced to distributed computation to tolerate
stragglers. The initial work was proposed in \cite{DBLP:conf/isit/2016, lee2018speeding} that they aim at
large-scale matrix multiplication. They encode the matrix to tolerate stragglers and design a coding
shuffling algorithm to reduce the data shuffling traffic. An improvement in \cite{maityrobust, maity2018robust}
is that they encode the second moment of data for the linear regression problem to reduce computational overhead
of encoding naive data. \cite{li2018polynomially, yu2018lagrange, ozfaturay2018speeding} utilize polynomial
interpolation to design the coded computation to tolerate more stragglers under the same workload compared
to traditional coding method. But different from our model, all these algorithms are only limited to the
linear model which couldn't be adopted by a broad of optimization problems.
For example, this strict condition cannot be satisfied by current DNN models. A general coding method named
gradient coding was proposed in \cite{tandon2017gradient}. Different from traditional works that encode the
data directly, they encode the gradients generated by optimization algorithm such that the linear model
constraint could be ruled out. Based on \cite{tandon2017gradient}, \cite{ye2018communication} proposes
reducing communication overhead by using coding method but further increases computing load incurred by
coding method. Besides, both their coding methods have not taken computing capacity of workers into account
causing the waste of computing resource. Though \cite{raviv2017gradient} and \cite{charles2017approximate}
aim at reducing computing load of coding method, they are at the cost of scarificing optimization accuracy.
Recognizing that, here in this paper we propose a heterogeneity-aware gradient coding method for general
optimization problem which not only takes computing capacity into account but also keeps accuracy of model.

\section{Problem Formulation}

\subsection{The Framework}
    Consider a typical distributed learning system, as illustrated in Fig.\ref{HeterogeneousCluster},
which consists of a master and a set of $m$ workers denoted by $\mathcal{W}=\{W_1,W_2,\ldots,W_m\}$.
A whole dataset $\mathcal{D}$ is divided into $k$ equal-sized data partitions,
denoted by $D_1,D_2,\ldots,D_k$, i.e., $\mathcal{D}=\{D_1,D_2,\ldots,D_k\}$. The partial gradient
over a data partition $D_i\in \mathcal{D}$ is denoted as $\mathbf{g}_i$, which can be obtained by
computation with $\mathcal{D}_i$. The whole task of distributed computation over this learning
system is to obtain the aggregated gradient as
\begin{equation*}
  \mathbf{g}=\sum_{i=1}^k \mathbf{g}_i.
\end{equation*}

\begin{figure}[htbp]
   \centering
   \includegraphics[width = 220pt, height = 150pt]{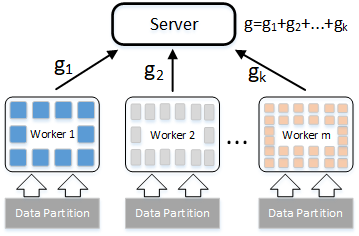}
   \caption{Distributed learning system with $m$ possible heterogeneous workers, where small
   rectangulars represent computing units. The main component of task is that server aggregates
   all patial gradients from workers.}
   \label{HeterogeneousCluster}
\end{figure}

A direct approach is to allocate different data partitions to different workers. Each worker computes
the partial gradients over the data partitions in hand, and then sends the summation of these partial
gradients to the master. After  collecting all the summations from the workers, the master can get
the aggregated gradient by summing up the summations. However, when there exists some straggler, the
computation latency could be significantly increased. Even worse, when some worker fails (e.g.,virtual
machine breaks down), the whole task cannot be completed. In order to tolerate stragglers/failures,
we consider the following general coding-based scheme. Initially, each worker $W_i$ is allocated
with a subset of data partitions $\mathcal{D}_i\subseteq \mathcal{D}$, where different $\mathcal{D}_i$
could be joint. Then $W_i$ computes all the corresponding gradients $\{\mathbf{g}_j\}_{j\in \mathcal{D}_i}$.
After this, $W_i$ encodes these gradients as $\mathbf{\tilde{g}}_i=e_i(\{\mathbf{g}_j\}_{j\in \mathcal{D}_i})$,
where $e_i$ is the encoding function of $W_i$, and sends $\tilde{\mathbf{g}}_i$ to the master.
After receiving enough results from some workers, say $\tilde{\mathcal{W}}\subseteq \mathcal{W}$, the
master recovers the desired aggregated gradient
$\mathbf{g}=h(\{\tilde{\mathbf{g}}_i\}_{i\in \tilde{\mathcal{W}}})$ immediately, where $h$ is
referred to as the decoding function.

\subsection{Gradient Coding Strategy}
Same as \cite{tandon2017gradient}, we consider linear encoding functions, i.e.,
$\tilde{\mathbf{g}}_i$ is a linear combination of $\mathbf{g}_j$, $j\in \mathcal{D}_i$.
Specifically, we can represent $\tilde{\mathbf{g}}_i$ as
\begin{equation*}
  \tilde{\mathbf{g}}_i=\mathbf{b}_i \cdot [\mathbf{g}_1,\mathbf{g}_2,\ldots,\mathbf{g}_k]^T,
\end{equation*}
where vector $\mathbf{b}_i\in \mathbb{R}^k$, and its support, denoted by $\text{supp}(\mathbf{b}_i)$,
which is the set of indices of non-zero entries of $\mathbf{b}_i$, satisfies that
$\text{supp}(\mathbf{b}_i)=\{j\mid D_j\in \mathcal{D}_i\}$, i.e., the indices of non-zero entries
of $\mathbf{b}_i$ show the allocation of data partitions to worker $W_i$.
Let $\mathbf{B}=[\mathbf{b}_1,\mathbf{b}_2,\ldots,\mathbf{b}_m ]^T\in \mathbb{R}^{m\times k}$,
which not only describes the allocation of data partition to each worker, but also represents
the encoding function of each worker. Henceforth, we will refer to $\mathbf{B}$ as a
\emph{gradient coding strategy}.

We seek gradient coding strategies that are robust to any $s$ stragglers with $s<m$.
Same as \cite{tandon2017gradient}, we assume that any straggler is a full straggler, i.e, it
can be arbitrarily slow to the extent of complete failure. Under this assumption, a sufficient
and necessary condition for a gradient coding strategy to be robust to any $s$ stragglers has
been shown in \cite{tandon2017gradient} as follows.
\begin{lemma}
\label{lem:C1}
  A gradient coding strategy
  $\mathbf{B}=[\mathbf{b}_1,\mathbf{b}_2,\ldots,\mathbf{b}_m ]^T\in \mathbb{R}^{m\times k}$
  is robust to any $s$ stragglers if and only if $\mathbf{B}$ satisfies the following condition:

  \noindent \textbf{(Condition 1):} for any subset $I\subseteq [m]$, $|I|=m-s$,
  \begin{equation}
    \mathbf{1}_{1\times k}\in \mathbf{span}(\{\mathbf{b}_i\mid i\in I\}),
  \end{equation}
  where $\mathbf{1}_{1\times k}$ is a all one vector, and $\mathbf{span}(\cdot)$ is the span of vectors.
\end{lemma}

    Given the coding strategy $\mathbf{B}$ that satisfies the condition (C1), the decoding strategy
$\mathbf{A} \in \mathbb{R}^{S \times m}$ could be correspondingly acheived for all $S$ stragglers patterns,
where $S=\binom{m}{s}$. Considering each row $\mathbf{a}_i$ of $\mathbf{A}$ denotes a specific scenario
of stragglers, master decodes $g$ by using coded gradients $\{\tilde{\mathbf{g}}_j\}_{j\in supp(\mathbf{a}_i)}$
sent by workers in $supp(\mathbf{a}_i)$. Accordingly, the decoding function can also be a linear combination as
    \begin{align*}
        h(\{\tilde{\mathbf{g}}_j\}_{j\in supp(\mathbf{a}_i)}) &= \sum_{j \in supp(\mathbf{a}_i)} \mathbf{a}_i(j) \tilde{\mathbf{g}}_j \\
        &= \mathbf{a}_i \mathbf{B} \cdot [\mathbf{g}_1,\mathbf{g}_2,\ldots,\mathbf{g}_k]^T
    \end{align*}
Hence, the decoding strategy $\mathbf{A}$ can be constructed by using
   \begin{equation}
   \label{eq:Aconstruct}
       \mathbf{A}_{Sm} \mathbf{B}_{mk} = \mathbf{1}_{Sk}
   \end{equation}
To reduce storage cost, the decoding matrix $\mathbf{A}$ could be partially stored specially for regular
stragglers. As to decoding functions $h(\{\tilde{\mathbf{g}}_j\}_{j\in supp(\mathbf{a}_i)})$ designed
for unregular stragglers, the decoding vectors $\mathbf{a}_i$ could solved in realtime in a complexity
of $O(mk^2)$. Note that the time for solving decoding vector usually can be ignored due to $m$ and $k$
are usually small numbers.

\subsection{Problem Formulation}

    Besides the tolerance of stragglers, we mainly concern about the computation time of the whole task.
We consider heterogeneous workers which have different computation capabilities. For each worker $W_i$,
let $c_i$ denote the number of partial gradients over data partitions that can be computed when $W_i$
is a non-straggler, which can be estimated by sampling. Thus, given a gradient coding strategy
$\mathbf{B}=[\mathbf{b}_1,\mathbf{b}_2,\ldots,\mathbf{b}_m ]^T$, the computation time of worker $W_i$,
denoted by $t_i$, is given by
\begin{equation*}
  t_i=\frac{{||\mathbf{b}_i||}_0}{c_i},
\end{equation*}
where ${||\mathbf{b}_i||}_0$ denotes the $\ell_0$-norm of $\mathbf{b}_i$, or equivalently, the
cardinality of $\text{supp}(\mathbf{b}_i)$. Without loss of generality, we assume that
$t_1\leq t_2\leq \cdots \leq t_m$.

Evidently, the computation time of the whole task under strategy $\mathbf{B}$ depends on which
workers are stragglers, referred to as straggler pattern. For a considered straggler pattern
$\mathcal{S}$, the computation time of the whole task under strategy
$\mathbf{B}$, denoted by $T(\mathbf{B},\mathcal{S})$ can be characterized as
\begin{equation*}
  T(\mathbf{B},\mathcal{S})=t_{j^*},
\end{equation*}
where $j^*$ is the minimum value of $j$ such that
\begin{equation*}
\mathbf{1}_{1\times k}\in \mathbf{span}(\{\mathbf{b}_i\}_{i\leq j, W_i\notin \mathcal{S}}).
\end{equation*}

For a gradient coding strategy $\mathbf{B}$ that can tolerate up to $s$ stragglers, we evaluate its
performance by the computation time of the whole task under $\mathbf{B}$ in the worst case, which
is denoted by $T(\mathbf{B})$ and is given by
\begin{equation}
  \label{eq:TDefinition}
  T(\mathbf{B})=\max_{\mathcal{S}\subset \mathcal{W}:|\mathcal{S}|\leq s} T(\mathbf{B},\mathcal{S}).
\end{equation}
Aiming at finding a gradient coding strategy with a best performance, we have the following
optimization problem:
\begin{IEEEeqnarray}{rCl}
   \label{OptimizationProblem}
  &\min\quad &T(\mathbf{B}) \nonumber\\
  & \text{s.t.} \quad & \mathbf{B}\text{ satisfies Condition 1}.
\end{IEEEeqnarray}

    For ease of reading, the main notations used in this paper are sumarized in the following
Table.\ref{tab:Symbol}
\begin{table}[htbp]
\centering
\caption{Symbols}
\label{tab:Symbol}
\begin{tabular}{ll}
    \toprule
    Symbol          &   Definition                            \\
    \midrule
    $m$             &   The number of worker                  \\
    $k$             &   The number of data partition          \\
    $s$             &   The number of stragglers              \\
    $W_i$           &   Worker $W_i$                          \\
    $n_i$           &   The number of data partitions in worker $W_i$ \\
    $c_i$           &   The throughput of worker i            \\
    $\mathbf{A}$    &   Decoding matrix                       \\
    $\mathbf{B}$    &   Coding matrix                         \\
    $\mathbf{1}$    &   Matrix with all elements being $1$    \\
    $[m]$           &   \{$1$,\ldots, $m$\}                   \\
    $supp(\mathbf{b})$  &   \{$i$ $\mid$ $v_i \neq 0$, $v_i$ is the element of vector $\bm{v}$\} \\
    $\mathcal{S}$   &   \{$i$ $\mid$ $i \in [m]$, $i$ is straggler \}       \\
    $\mathcal{D}$   &   The set of all data partitions                      \\
    $\mathcal{W}$   &   The set of all workers                              \\
    $\mathcal{G}$   &   Group composed of workers                           \\
    $\mathcal{P}$   &   Groups set composed of groups                       \\
    \bottomrule
\end{tabular}
\end{table}

\section{Heterogeneity-aware Gradient Coding Strategy}
In this section, we will show our coding scheme for heterogeneous distributed system detailly.
Firstly, we specify how to design the support of $\mathbf{B}$ with the considering of load
balance and stragglers tolerance. We solve this by designing an heterogneity-aware data
allocation scheme. After that, the construction process of $\mathbf{B}$ is elaborated which
is the key for accurate decoding. Finaly, we show that our coding strategy is optimal to
problem (\ref{OptimizationProblem}).

\subsection{The Design}
We first show how to allocate data partitions to the workers, which gives the support structure
of $\mathbf{B}$, i.e., the positions of non-zero elements in $\mathbf{B}$.

In order to tolerate $s$ stragglers, each data partition $D_i$ has to be assigned to at least
$s+1$ workers to compute $\mathbf{g}_i$. In our design, $D_i$ is copied exactly $s+1$ times,
and there are in total $k(s+1)$ copies of data partitions, i.e., $\sum_{i=1}^m n_i=k(s+1)$,
where $n_i$ is the number of data partitions assigned to worker $W_i$. For load balancing,
we set $n_i$ to be proportional to $c_i$, the computation rate of $W_i$. Hence, we have
\begin{equation}
  n_i=k(s+1)\cdot \frac{c_i}{\sum_{j=1}^m c_j}.
\end{equation}
Without loss of generality, here we assume that $k(s+1)\cdot \frac{c_i}{\sum_{j=1}^m c_j}$
is an integer, and $n_i\leq k$.

Once $n_i$ are fixed, we assign the total $k(s+1)$ copies of data partitions to the workers
in a cyclic manner. Specifically, the set of  data partitions assigned to worker $W_i$ are
given as
\begin{equation}
  \label{eq:DataAllocation}
  \mathcal{D}_i=\{D_{(n_i'+1)\text{mod }k},D_{(n_i'+2)\text{mod }k},\ldots,D_{(n_i'+n_i)\text{mod }k}\}.
\end{equation}
where $n_i'=\sum_{j=1}^{i=1} n_j$. It is straightforward to see that, for each $D_i$, there
are exact $s+1$ copies assigned to $s+1$ different workers. By denoting $\star$ as non-zero entry,
the support structure of worker $W_i$ is $supp(\mathbf{b}_i)=[b_1,b_2,\cdots, b_k]$,
where $b_j=\star$ if $D_j \in \mathcal{D}_i$ else $b_j=0$, and the support sutructure of
$\mathbf{B}_{m \times k}$ can be written as

\begin{equation}
    \label{mtx:suppB}
    supp(\mathbf{B}_{m \times k})= [\mathbf{b}_1, \mathbf{b}_2, \cdots, \mathbf{b}_m]^T
\end{equation}

\begin{example}
As an example, consider a $5$-workers system with normalized sampling throughput as $\mathbf{c}=[1,2,3,4,4]$.
If there is $1$ straggler, we could allocate data partitions and determine suppoprt structure of
$\mathbf{B}$ as

\begin{equation*}
    \label{example:suppB}
    supp(\mathbf{B}_{5 \times 7})= \\
    \begin{bmatrix}
        \star & 0      & 0     & 0      & 0      & 0     & 0       \\
        0     & \star  & \star & 0      & 0      & 0     & 0       \\
        0     & 0      & 0     & \star  & \star  & \star & 0       \\
        \star & \star  & \star & 0      & 0      & 0     & \star   \\
        0     & 0      & 0     &\star   & \star  & \star & \star
    \end{bmatrix}
\end{equation*}
\end{example}

Given the support structure of $\mathbf{B}$, we now introduce how to construct $\mathbf{B}$ such that
it can satisfy the condition (C1). In our construction, an auxiliary matrix
$\mathbf{C}\subseteq \mathbb{R}^{(s+1)\times m}$ is introduced, which satisfies the following properties.
\begin{itemize}
  \item (P1): any $s+1$ columns of $\mathbf{C}$ is linearly independent.
  \item (P2): for any submatrix $\mathbf{C}'$ composed by $s$ columns of $\mathbf{C}$ and any non-zero
  vector $\boldsymbol{\lambda} =(\lambda_1,\ldots,\lambda_{s+1})\in \mathbb{R}^{s+1}$ such
  that $\boldsymbol{\lambda} \mathbf{C}'=\mathbf{0}_{1\times s}$, $\sum_{i=1}^{s+1}\lambda_{i}\neq 0$.
\end{itemize}
The usefulness of such a $\mathbf{C}$ is revealed by the following result.

\begin{lemma}
\label{lem:Cproperty}
  For a matrix $\mathbf{C}\subseteq \mathbb{R}^{(s+1)\times m}$ having properties (P1) and (P2),
  there exists a matrix $\mathbf{B}\subseteq \mathbf{R}^{m\times k}$ with a support structure of (\ref{mtx:suppB})
  such that $\mathbf{C}\mathbf{B}=\mathbf{1}_{(s+1)\times k}$ and $\mathbf{B}$ satisfies condition (C1).
\end{lemma}

\begin{proof}
  Our proof proceeds as follows. First, we construct a matrix $\mathbf{B}\subseteq \mathbb{R}^{m\times k}$
with a support structure of (\ref{mtx:suppB}) such that $\mathbf{C}\mathbf{B}=\mathbf{1}_{(s+1)\times k}$.
Then, we show that $\mathbf{B}$ satisfies condition (C1).

For each $i=1,2,\ldots,k$, let $\mathbf{C}_i$ be the submatrix of $\mathbf{C}$ by deleting
all the $j$-th columns where the $j$-th element of the $i$-th column of the support structure (\ref{mtx:suppB})
is zero. Since each column of the support structure (\ref{mtx:suppB}) has $s+1$ non-zero elements,
$\mathbf{C}_i$ has $s+1$ columns which are linearly independent according to property (P1). Therefore,
$\mathbf{C}_i$ is non-singular, and has an inverse which is denoted by $\mathbf{C}_i^{-1}$.
Let
\begin{equation*}
  \mathbf{d}_i'=\mathbf{C}_i^{-1}\mathbf{1}_{(s+1)\times 1},
\end{equation*}
and $\mathbf{B}$ be the matrix formed by embedding each $\mathbf{d}_i'$, $i=1,2,\ldots,k$ into the $i$-th
column of the support structure (\ref{mtx:suppB}). The embedding process is to assign each value in
$\mathbf{d}_i'$ to $b_i$ according to the position presented in $supp(\mathbf{b}_i)$. Evidently,
$\mathbf{C}\mathbf{B}=\mathbf{1}_{(s+1)\times k}$.

Next we show that the constructed $\mathbf{B}$ satisfies condition C1. Let
$\mathbf{b}_1,\mathbf{b}_2,\ldots,\mathbf{b}_m$ be the rows of $\mathbf{B}$.
Consider an arbitrary subset $I\subseteq [m]$ such that $|I|=m-s$. Let $\mathbf{C}_{\bar{I}}$
be the submatrix composed by all the $j$-th columns of $\mathbf{C}$ where $j\notin I$.
Since $\mathbf{C}_{\bar{I}}$ has $s$ columns while it has $s+1$ rows, there exists some non-zero
vector $\boldsymbol{\lambda}=(\lambda_1,\lambda_2,\ldots,\lambda_{s+1})\in \mathbf{R}^{s+1}$ such
that $\boldsymbol{\lambda}\mathbf{C}_{\bar{I}}=\mathbf{0}_{1\times s}$. Since $\mathbf{C}$ satisfies
property (P2), we have $\sum_{i=1}^{s+1} \lambda_i\neq 0$. Hence,
\begin{equation*}
  \left(\frac{1}{\sum_{i=1}^{s+1}\lambda_i}\boldsymbol{\lambda} \mathbf{C}\right)\mathbf{B}=\frac{1}{\sum_{i=1}^{s+1}\lambda_i}\boldsymbol{\lambda}(\mathbf{C}\mathbf{B})=\mathbf{1}_{1\times k}.
\end{equation*}
Note that for  $j\notin I$, the $j$-th entry of the vector $\frac{1}{\sum_{i=1}^{s+1}\lambda_i}\boldsymbol{\lambda} \mathbf{C}$ is equal to 0 since $\boldsymbol{\lambda}\mathbf{C}_{\bar{I}}=\mathbf{0}_{1\times s}$. We then have that  $\mathbf{1}_{1\times k}$ belongs to the span of $\{\mathbf{b_j}\}_{j\in I}$. Therefore, $\mathbf{B}$ satisfies condition (C1). The proof is accomplished.
\end{proof}

In the proof of Lemma~\ref{lem:Cproperty}, we give a construction method of $\mathbf{B}$ with
desired properties if we have a matrix $\mathbf{C}$ satisfying properties (P1) and (P2).
Hence, all we need now is to construct such a matrix $\mathbf{C}$. In the following, we
show that a random choice of $\mathbf{C}$ suffices where each entry of $\mathbf{C}$ is
chosen from the interval $(0,1)$ independently and uniformly at random.
\begin{lemma}
    \label{lem:CRandom}
      For a matrix $\mathbf{C}\subseteq \mathbb{R}^{(s+1)\times m}$ where  each entry of
      $\mathbf{C}$ is chosen from the interval $(0,1)$ independently and uniformly at random, then
      $\mathbf{C}$ satisfies both properties of (P1) and (P2) with probability 1.
    \end{lemma}
\begin{proof}
  It has been shown in \cite{tandon2017gradient} that $\mathbf{C}$ satisfies property (P1) with
  probability 1. So we only need to show that $\mathbf{C}$ satisfies (P2) with probability 1.

  Consider any submatrix $\mathbf{C}'$ composed by $s$ columns of $\mathbf{C}$.
  Let $\mathbf{c}_1,\ldots,\mathbf{c}_{s+1}$ be the rows of $\mathbf{C}'$. Without loss of
  generality, we assume that the values of $\mathbf{c}_1,\ldots,\mathbf{c}_s$ have been exposed
  and they are independent which holds with probability 1, so that we focus on the randomness
  of $\mathbf{c}_{s+1}$.  Let
  $\boldsymbol{\lambda}'(\mathbf{c}_{s+1})=(\lambda_1'(\mathbf{c}_{s+1}),\ldots,\lambda_s'(\mathbf{c}_{s+1}))$,
  which is unique, such that
  \begin{equation*}
    \mathbf{c}_{s+1}=\lambda_1'(\mathbf{c}_{s+1})\mathbf{c}_1+\lambda_2'(\mathbf{c}_{s+1})\mathbf{c}_2+\cdots+\lambda_s'(\mathbf{c}_{s+1})\mathbf{c}_s.
  \end{equation*}
  We can check that $\boldsymbol{\lambda}'(\mathbf{c}_{s+1})$ is a continuous multivariate random
  variable. So the probability of $\sum_{i=1}^s \lambda_s'(\mathbf{c}_{s+1})\neq 1$ is 1.  On the
  other hand, if $\sum_{i=1}^s \lambda_s'(\mathbf{c}_{s+1})\neq 1$, then for any non-zero vector
  $\boldsymbol{\lambda}=(\lambda_1,\ldots,\lambda_{s+1})\in \mathbb{R}^{s+1}$ such that
  $\boldsymbol{\lambda} \mathbf{C}'=\mathbf{0}_{1\times s}$, $\lambda_{s+1}\neq 0$ and
  $\boldsymbol{\lambda}'(\mathbf{c}_{s+1})=\left(\frac{\lambda_1}{\lambda_{s+1}},\frac{\lambda_2}{\lambda_{s+1}},\ldots,\frac{\lambda_s}{\lambda_{s+1}}\right)$.
   Therefore, $\sum_{i=1}^{s+1}\lambda_i\neq 0$. This implies that the property (P2) restricted
   to the $\mathbf{C}'$ holds with probability 1. Since there are $\binom{s}{m}$ such $\mathbf{C}'$,
   taking a union bound over them shows that property (P2) holds with probability 1.
\end{proof}

The algorithm for constructing $\mathbf{B}$ is given in Alg.\ref{alg:BConstruct}
\begin{algorithm}[htbp]
    \caption{Heter-aware Coding Scheme}
    \label{alg:BConstruct}
    \hspace*{0.02in} {\bf Input: }
        $k$, $supp(\mathbf{B})$ \\
    \hspace*{0.02in} {\bf Output: $\mathbf{B}$}
    \begin{algorithmic}[1]
        \State initialize $\mathbf{B} = zeros(m, k)$
        \For{$i$ in $[s+1]$}
            \For{$j$ in $[m]$}
                \State $\mathbf{C}(i)(j) = random(0,1)$
            \EndFor
        \EndFor
        \For{$i$ in $[k]$}
            \State $\mathbf{b}=zeros(m,1)$
            \State $filter = supp(\mathbf{B})(i)$
            \For{$j$ in $[s+1]$}
                \For{$l$ in $[s+1]$}
                    \State $\mathbf{C}_i(j)(l)=\mathbf{C}(j)(filter(l))$
                \EndFor
            \EndFor
            \State $\mathbf{d}_i' = \mathbf{C}_i^{-1} \mathbf{1}_{(s+1) \times 1}$
            \For{$j$ in $[s+1]$}
                \State $\mathbf{b}(filter(j)) = \mathbf{d}_i'(j)$
            \EndFor
            \State $\mathbf{B}_i = \mathbf{b}$
        \EndFor
        \State \textbf{return} $\mathbf{B}$
    \end{algorithmic}
\end{algorithm}

As a consequence of Lemma~\ref{lem:C1}, Lemma~\ref{lem:Cproperty} and Lemma~\ref{lem:CRandom},
we have the following theorem immediately.
\begin{theorem}
    \label{tm:RobustnessB}
    The matrix $\mathbf{B}$ constructed by Alg. is robust to any $s$ stragglers with probability 1.
\end{theorem}

\subsection{Optimality}

\begin{theorem}
  The gradient coding strategy $\mathbf{B}$ constructed by Alg.\ref{alg:BConstruct} is an optimal
  solution to problem (\ref{OptimizationProblem}) with probability $1$.
\end{theorem}
\begin{proof}
  Let $\mathbf{B}^*$ be an optimal gradient coding strategy. Let $\mathbf{b}_i^*$ be the $i$-th
  row of $\mathbf{B}^*$. If there exists some $i$ such that
  $T(\mathbf{B}^*)<\frac{||\mathbf{b}_i^*||_0}{c_i}$, then according to the definition of
  $T(\mathbf{B}^*)$ (c.f. Eq. (\ref{eq:TDefinition})), the result of worker $W_i$ is useless
  for earliest successful decoding whatever the straggler pattern is. Hence, we can remove
  the assignment of data partitions to worker $W_i$, which does not affect the straggler
  tolerance and the computation time of the whole task. In other words, this is still
  an optimal gradient coding strategy. Hence, we can conclude that there exists an optimal
  gradient coding strategy
  $\tilde{\mathbf{B}}^*$ such that
   \begin{equation*}
     T(\tilde{\mathbf{B}}^*)\geq \frac{||\tilde{\mathbf{b}}_i^*||_0}{c_i}, \text{ for any }i=1,2,\ldots,m,
   \end{equation*}
   where $\tilde{\mathbf{b}}_i$ is the $i$-th row of $\tilde{\mathbf{B}}^*$. Now we have
   \begin{equation*}
     T(\tilde{\mathbf{B}}^*)\geq \frac{\sum_{i=1}^m ||\tilde{\mathbf{b}}_i^*||_0}{\sum_{i=1}^m c_i}.
   \end{equation*}
   On the other hand, in order to tolerate $s$ straggler, each data partition has to be
   assigned to at least $s+1$ workers. This implies that
   \begin{equation*}
     \sum_{i=1}^k ||\tilde{\mathbf{b}}_i^*||_0\geq (s+1)k.
   \end{equation*}
   Hence,
   \begin{equation*}
     T(\tilde{\mathbf{B}}^*)\geq \frac{(s+1)k}{\sum_{i=1}^m c_i}.
   \end{equation*}
For our construction $\mathbf{B}$, we can see that every worker completes its local task
in $\frac{(s+1)k}{\sum_{i=1}^m c_i}$ time according to Eq. (2). Hence,
$T(B)=\frac{(s+1)k}{\sum_{i=1}^m c_i}$, which implies that $\mathbf{B}$ is optimal.
\end{proof}

\section{Group-based Coding Scheme}
    Based on the sampling throughput $c_i$ of each worker $W_i$, we have proposed an
optimal solution for problem (\ref{OptimizationProblem}) on the above section. However,
$c_i$ in practical system is hard to be measured exactly because of tiny fluctuation
in runtime. This leads to that coding scheme could hardly acheive optimal.
In fact, we could further improve the performance by reducing the number of
workers $|\mathcal{A}|$ needed by recovering gradient. This is because (1)
if $\mathcal{A}_1 \subset \mathcal{A}_2$, then $T_{\mathcal{A}_1} \leq T_{\mathcal{A}_2}$
with probability $1$ where $T_\mathcal{A}$ is the recovering time from active
workers $\mathcal{A}$ and (2) from lemma \ref{lem:Cproperty}, we could directly
conclude that recovering gradient from $\mathbf{B}$ constructed by
Alg.\ref{alg:BConstruct} needs $m-s$ workers given $s$ stragglers.
In the follows, we show that $|\mathcal{A}|$ could be reduced by finding groups, where
a group consists of at most $m-s$ workers and can be used to recover gradient.
Denote a group as $\mathcal{G}$, then the following conditions are desired to satisfy
requirement
\begin{itemize}
    \item ($\star$): for all workers $W_{z_i} \in \mathcal{G}, i=1,2,\ldots,p$, their sets of
    data partitions satisfy
        \begin{equation*}
            \bigcap_{i=1}^p \mathcal{D}_{z_i} = \emptyset, \quad
            \bigcup_{i=1}^p \mathcal{D}_{z_i} = \mathcal{D}
        \end{equation*}
    \item ($\star\star$): for all groups $\mathcal{G}_i$ in $\mathbf{B}$, $i=1,2,\ldots,P$,
        \begin{equation*}
            \bigcap_{i=1}^P \mathcal{G}_i = \emptyset
        \end{equation*}
\end{itemize}

\begin{algorithm}[htbp]
    \caption{Find Groups}
    \label{alg:FindGroups}
    \hspace*{0.02in} {\bf Input: $\mathcal{D}, \mathcal{W}$} \\
    \hspace*{0.02in} {\bf Output: $\mathcal{P}$}
    \begin{algorithmic}[1]
    \State $\mathcal{P}$ = FindAllGroups($\mathcal{D}, \mathcal{W}$)
    \State $\mathcal{P}$ = PruneGroups($\mathcal{P}$)
    \State \textbf{return} $\mathcal{P}$
    \\
    \Function{FindAllGroups}{$\mathcal{D}, \mathcal{W}$}
        \State initialize groups set $\mathcal{P} = \{\}$
        \State $\mathcal{W}_c = \mathcal{W}.clone()$
        \For{$W_i$ in $\mathcal{W}_c$}
            \State $\mathcal{W} -= W_i$
            \If{$\mathcal{D}_i \subset \mathcal{D}$}
                \State $\mathcal{D}_r = \mathcal{D} - \mathcal{D}_i$
                \State $\mathcal{P}_s$ = FindAllGroups($\mathcal{D}_r, \mathcal{W}, \mathcal{D}_i$)
                \For{subgroup $\mathcal{G}$ in $\mathcal{P}_s$}
                    \State $\mathcal{G} +=W_i$
                    \State $\mathcal{P} += \mathcal{G}$
                \EndFor
            \ElsIf{$\mathcal{D}_i = \mathcal{D}$}
                \State $\mathcal{G} = \{W_i\}$
                \State $\mathcal{P} += \mathcal{G}$
            \Else
                \State pass;
            \EndIf
        \EndFor
        \State \textbf{return} $\mathcal{P}$
    \EndFunction
    \\
    \Function{PruneGroups}{$\mathcal{P}$}
        \While{$\mathcal{P}$ doesn't satisfy condition ($\star\star$)}
            \State find $\mathcal{G}$ that intersects most groups
            \State $\mathcal{P} -= \mathcal{G}$
        \EndWhile
        \State \textbf{return} $\mathcal{P}$
    \EndFunction
    \end{algorithmic}
\end{algorithm}

As shown in Alg.\ref{alg:FindGroups}, all groups are found in a recursive way to satisfy
condition ($\star$), and then several groups are pruned to satisfy condition ($\star\star$). After
finding groups, we just set non-zero elements in $\mathbf{B}$ corresponding workers $W_j$
in groups $\mathcal{G}_i$ to be $1$. Besides, consider a group set
$E = \{ W_j \mid W_j \mathcal{G}_i, i=1,2,\ldots, P\}$. Let $\mathbf{B}_{\mathbf{\bar{E}}}$
be the submatrix composed by all the $j$-th rows of $\mathbf{B}$ where worker $W_j \notin E$.
Obviously, $\mathbf{B}$ can be constructed as long as the submatrix
$\mathbf{B}_{\mathbf{\bar{E}}}$ is solved. $\mathbf{B}_{\mathbf{\bar{I}}}$ can be
constructed by using Alg.\ref{alg:BConstruct} under $s=m-P$ stragglers.

A little different from decoding function for $\mathbf{B}$ constructed by
Alg.\ref{alg:BConstruct}, the decoding matrix $\mathbf{A}$ are constructed for workers
in groups and workers not in groups separately. For workers in each group $\mathcal{G}_i$,
we design each corresponding decoding vector $\mathbf{a}_i \in \mathbb{R}^m$ as
$\mathbf{a}_i = [\mathbf{1}_{\mathcal{G}_i}(W_1),\cdots,\mathbf{1}_{\mathcal{G}_i}(W_m)]$,
where $\mathbf{1}_{\mathcal{G}_i}$ is the indicator function. Obviously, we have
    \begin{equation}
        \label{eq:OneGroupDecoding}
        \mathbf{a}_i \mathbf{B} = \mathbf{1}, \quad, \|\mathbf{a}_i\|_0 \leq m-s
    \end{equation}

\begin{figure*}[htbp]
    \centering
    \subfloat[s=1 Straggler]{
        \label{fig:fig1_a}
        \begin{minipage}[t]{240pt}
            \centering
            \includegraphics[width = 240pt, height = 200pt]{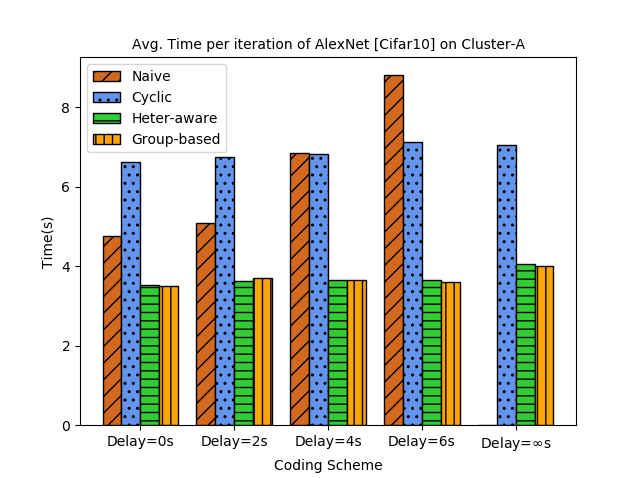}
        \end{minipage}
    }
    \subfloat[s=2 Stragglers]{
        \label{fig:fig1_b}
        \begin{minipage}[t]{240pt}
            \centering
            \includegraphics[width = 240pt, height = 200pt]{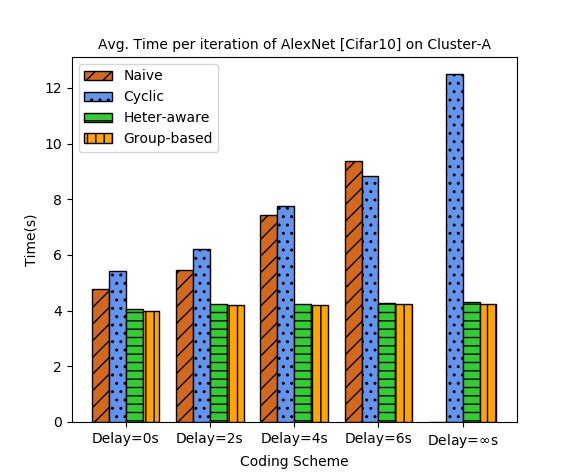}
        \end{minipage}
    }
    \caption{Avg. time per iteration of different coding schemes running on Cluster-A with
    $s=1$ and $s=2$ stragglers. The stragglers are created artificially by adding delay to the workers.
    The results show that our proposed \textit{heter-aware} and \textit{group-based} gradient coding
    scheme performs best without regarding to the delays.}
    \label{fig:fig1}
\end{figure*}

Consequently, a decoding submatrix denoted by $\mathbf{A}_1$ is composed by decoding vectors
for all groups. As to workers not in groups, the decoding submatrix $\mathbf{A}_2$ is solved
by $\mathbf{B}_{\mathbf{\bar{E}}}$ according to (\ref{eq:Aconstruct}).

The algorithm for constructing $\mathbf{B}$ and solving decoding matrix $\mathbf{A}$ is shown
in Alg.\ref{alg:GroupBasedCoding}.

According to Theorem \ref{tm:RobustnessB}, we could have the following theorem.
\begin{theorem}
    \label{tm:GroupRobustness}
        The matrix $\mathbf{B}$ constructed by Alg.\ref{alg:GroupBasedCoding} is
    robust to any $s$ stragglers with probabilty $1$.
\end{theorem}
\begin{proof}
   According to both condition ($\star$) and ($\star\star$) of groups and
Theorem \ref{tm:RobustnessB}, we could easily know that
$\mathbf{B}_{\mathbf{\bar{E}}}$ are robust to $s-P$ stragglers with
probability $1$. Besides, all $P$ groups can be used to recover gradient as
showed in \label{eq:OneGroupDecoding}. Hence, $\mathbf{B}$ is robust to $s$
stragglers with probabilty $1$.
\end{proof}

From this theorem, we could know that the group-based coding scheme is also an optimal
solution for problem (\ref{OptimizationProblem}) for that the computation time of each
active worker is the same like of worker in coding scheme \ref{alg:BConstruct} under
deterministic situation.

\begin{example}
    An example is shown as in the following support structure $\mathbf{B}_{7 \times 4}$ of $7$
workers. There are three groups, $\mathcal{G}_1$ including workers $W_1, W_2, W_3$,
$\mathcal{G}_2$ including workers $W_3, W_4$ and $\mathcal{G}_3$ including workers $W_2, W_5$.
Laterly, system prunes $\mathcal{G}_1$ to satisfy condition ($\star\star$). For constructing
$\mathbf{B}$, all entries of workers $W_2, W_3, W_4, W_5$ that in groups are set to be $1$,
and the remained entries of $\mathbf{B}$ and are solved by using Alg.\ref{alg:BConstruct}.

\begin{equation*}
    supp(\mathbf{B}_{7 \times 4})=
    \begin{bmatrix}
        \star & \star & 0 & 0\\
        0 & 0 & \star & 0  \\
        0 & 0 &  0 & \star \\
        \star & \star & \star & 0 \\
        \star & \star & 0 & \star \\
        \star & 0 & \star & \star \\
        0 & \star & \star & \star
    \end{bmatrix}, \quad
    \mathbf{B}_{7 \times 4}=
    \begin{bmatrix}
        \star & \star & 0 & 0\\
        0 & 0 &  1 & 0  \\
        0 & 0 &  0 & 1 \\
        1 & 1 &  1 & 0 \\
        1 & 1 &  0 & 1 \\
        \star & 0 & \star & \star \\
        0 & \star & \star & \star
    \end{bmatrix}
\end{equation*}
\end{example}

\begin{algorithm}[htbp]
    \caption{Group-Detection Coding Scheme}
    \label{alg:GroupBasedCoding}
    \hspace*{0.02in} {\bf Input: }
        $supp(\mathbf{B}), \mathcal{P}$ \\
    \hspace*{0.02in} {\bf Output: }
        $\mathbf{A}, \mathbf{B}$
    \begin{algorithmic}[1]
        \State initialize $\mathbf{A}_1=[\ ]$
        \State $\mathbf{B}_{E} = 1$
        \For{$\mathcal{G}_i$ in $\mathcal{P}$}
            \State $\mathbf{a}_i = [\mathbf{1}_{\mathcal{G}_i}(W_1),\cdots,\mathbf{1}_{\mathcal{G}_i}(W_m)]$
            \State $\mathbf{A}_1.append(\mathbf{a}_i)$
        \EndFor
        \State solve $\mathbf{B}_{\bar{E}}$ via Alg.\ref{alg:BConstruct}
        \State solve $\mathbf{A}_2$ by $\mathbf{B}_{\bar{E}}$
        \State $\mathbf{A}$ = merge($\mathbf{A}_1, \mathbf{A}_2$)
        \State $\mathbf{B}$ = merge($\mathbf{B}_{E}, \mathbf{B}_{\bar{E}}$)
        \State \textbf{return} $\mathbf{A}, \mathbf{B}$
    \end{algorithmic}
\end{algorithm}

\section{Performance Evaluations}
In this section, experiments are presented to show the results of our coding scheme.
Our coding scheme was mainly compared to two schemes: 1) \textit{Naive} scheme.
In \textit{naive} scheme, the whole dataset was divided uniformly on each worker and server
makes an update step by waiting for the completion of all workers. 2) \textit{Cyclic} coding
scheme \cite{tandon2017gradient}. \textit{Cyclic} coding scheme uniformly divides the dataset
into $m$ data partitions and makes $s+1$ copies of each data partition, and each worker
computes $s+1$ data partitions. We didn't implement fractional repetition scheme
and partial coding scheme in \cite{tandon2017gradient}, because fractional repetition scheme
not only has a great limitation that requires that the number of worker $m$ is divisible by
$s+1$ but also its performance is comparable to cyclic coding scheme and as to partial coding
scheme, it a strong assumption that the slowest worker is at most $\alpha$ slower than the
fastest worker causing that it is unable to tolerate corrupted workers.

\begin{figure*}[htbp]
    \centering
    \subfloat[Cluster-B]{
        \label{fig:fig2_a}
        \begin{minipage}[t]{160pt}
            \centering
            \includegraphics[width = 180pt, height=160pt]{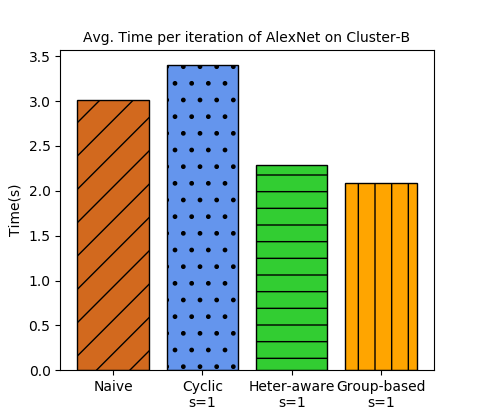}
        \end{minipage}
    }
    \subfloat[Cluster-C]{
        \label{fig:fig2_b}
        \begin{minipage}[t]{160pt}
            \centering
            \includegraphics[width = 180pt, height=160pt]{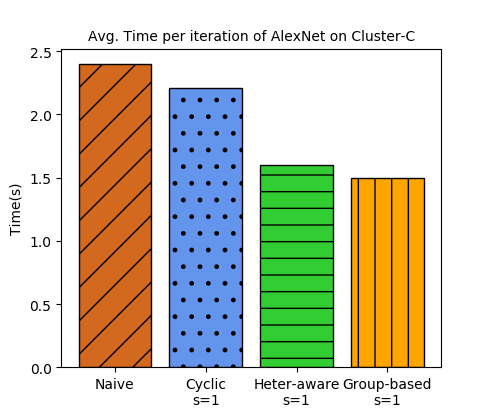}
        \end{minipage}
    }
    \subfloat[Cluster-D]{
        \label{fig:fig2_c}
        \begin{minipage}[t]{160pt}
            \centering
            \includegraphics[width = 180pt, height=160pt]{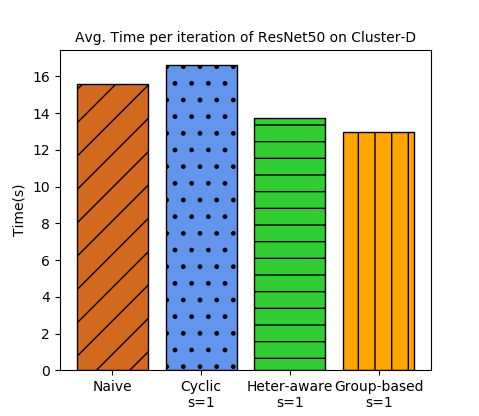}
        \end{minipage}
    }
    \caption{Avg.time per iteration on different clusters. Our coding schemes perform best in
    all clusters with different configurations.}
    \label{fig:fig2}
\end{figure*}

\begin{table}[htbp]
    \centering
    \caption{Cluster Configurations}
    \label{ClusterConf}
    \begin{tabular}{|c|c|c|c|c|}
        \hline
        number of vCPUs  & Cluster-A & Cluster-B & Cluster-C & Cluster-D \\
        \hline
        2-vCPUs          & 2         &  2        & 1         &  0 \\
        \hline
        4-vCPUs          & 2         &  4        & 4         &  4 \\
        \hline
        8-vCPUs          & 3         &  8        & 10         & 20 \\
        \hline
        12-vCPUs         & 1         &  0        & 12         & 18 \\
        \hline
        16-vCPUs         & 0         &  2        & 5         &  16 \\
        \hline
    \end{tabular}
\end{table}

\textbf{Experiment Setup.} Based on QingCloud \cite{QingCloud}, we make evaluations on
various heterogeneous clusters with different scales ranging from $8$ workers to $48$ workers.
We design four clusters including \textbf{Cluster-A}, \textbf{Cluster-B}, \textbf{Cluster-C} and
\textbf{Cluster-D} as shown in Table \ref{ClusterConf}. Such design mainly aims to cover various
scales and heterogeneity of cluster to show the generality of our coding scheme. The instance
type is performance type, and operating system of all the nodes is 64-bit Centos7.1.
PyTorch \cite{pytorch} is adopted as the platform.

\textbf{Workload.} Two typical image classification datasets Cifar10 \cite{krizhevsky2012imagenet}
and ImageNet \cite{deng2009imagenet} are adopted. Cifar10 is composed of $50,000$ $32 \times 32$
training images on which we train AlexNet \cite{krizhevsky2012imagenet}, and ImageNet consists
of over $1$ million images on which we train ResNet34 \cite{he2016deep}.

\textbf{Metrics.} System efficiency is measured by running time to show the overall efficiency of
distributed learning system. It consists of statistical efficiency and hardware efficiency.
Statistical efficiency measures the convergence rate of the learning algorithm can be shown by
learning curve. Hardware efficiency is a metric that represents the efficiency effcient CPU
resource usage.

\subsection{Experimental Results}
\subsubsection{Robustness to Stragglers}
    By simulating faults, we add extra delay to any $s$ random workers on \textbf{Cluster-A} to
show both the performance improvement and the ability of straggler tolerance of our coding scheme. We
artificially generate $1$ stragglers and $2$ stragglers as shown in Fig.\ref{fig:fig1_a} and
Fig.\ref{fig:fig1_b}. As expected, running time of \textit{naive} increases with the increasing
of delay and could not normally run as workers take place faults. Correspondingly, all coding schemes
are designed for $1$ straggler in Fig.\ref{fig:fig1_a} and for $2$ stragglers in Fig.\ref{fig:fig1_b}.
Different from \textit{naive} distributed learning algorithm, \textit{cyclic} algorithm could tolerate
stragglers that the running time changes little to different delays as shown in Fig.\ref{fig:fig1_a}.
However, the running time of \textit{cyclic} algorithm also increases with the increasing of delay.
This is mainly because the performance of \textit{cyclic} is mainly limited to workers with
low-computing capacity, and it approaches the performance of low-computing workers as delay increases
until reaches the lower bound as delay is infinite (faults take place). Compared to these two
distributed algorithm, both our \textit{heter-aware} coding scheme and \textit{group-based} coding
scheme are all robust to stragglers that the running time keeps almost unchanged as shown in
Fig.\ref{fig:fig1_a} and Fig.\ref{fig:fig1_b}. When the fault takes place, our \textit{heter-aware}
coding scheme even acheives $3\times$ speedup compared to \textit{cyclic} coding scheme because of
high computing resource usage.

\begin{figure}[htbp]
    \centering
    \label{fig:fig3_a}
    \includegraphics[width = 240pt, height=200pt]{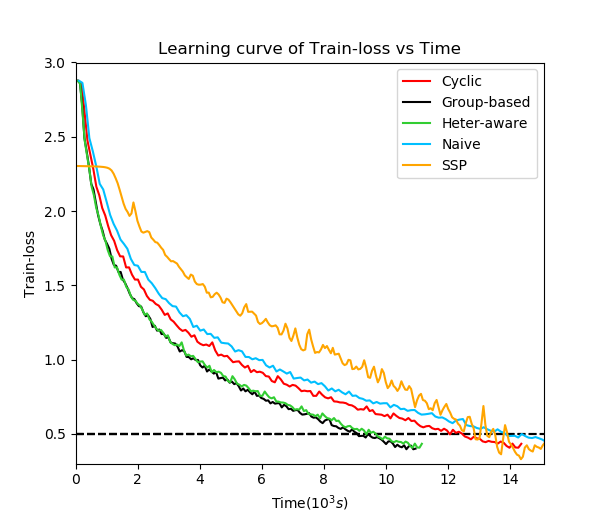}
    \caption{Training Loss curve of different learning schemes on Cluster-C. \textit{Group-based}
    coding scheme has the best convergence efficiency, and then \textit{heter-aware} coding
    scheme. \textit{Cyclic} coding scheme could only have a little better efficiency than
    \textit{Naive} learning method due to insufficient workload allocation. SSP performs
    worst in such heterogeneous setting due to consistent straggler and poor convergence rate.
    }
    \label{fig:fig3}
\end{figure}

\subsubsection{Efficiency under different clusters}
    To show generality and efficiency of our coding scheme, we extend experiments to a large range of
clusters with different scales and computing configurations as Cluster-B, Cluster-C and Cluster-D.
The results are shown as in Fig.\ref{fig:fig2}. Obviously, \textit{heter-aware} and \textit{group-based}
coding scheme acheive better performance than the other methods on each cluster of different
configurations. On the other side, traditional \textit{cyclic} coding scheme even makes performance
worse for that it aggreggates the straggler problem by allocating equivalent workload to each worker
with different computing capacity.

    Besides, one most notable advantage of coding based methods is that they have better
statistical efficiency by using BSP. This is not true in asynchronous  learning
algorithm, which is deeply discussed in \cite{chen2016revisiting}. We here
validate the efficiency of our learning method compared to SSP, a notablely effcient
asynchronous distributed learning algorithm. The result is shown as in Fig.\ref{fig:fig3}.
Due to heterogeneous computing capacity of workers, SSP will in fact easily reach the
staleness threshold nearly every step causing that the synchronization overhead is similar to
\textit{Naive} BSP learning algorithm. Besides, master receives unbalanced contributions from
different data parts to the update of parameters due to the dicrepancy of workers causing that
SSP has a lower convergence rate compared BSP. Consequently, our coding scheme converges
smoother and faster than SSP as shown in Fig.\ref{fig:fig3}.

\begin{figure}[htbp]
    \centering
    \label{fig:fig4_a}
    \includegraphics[width = 240pt, height=200pt]{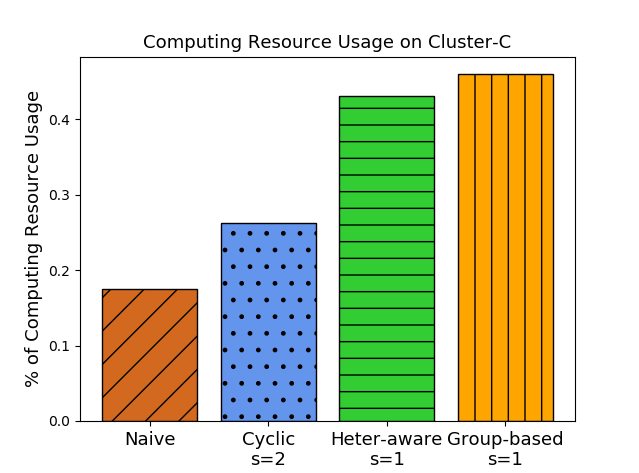}
    \caption{Computing resource usage of different coding schemes. Computing resource usage of
    \textit{group-based} coding scheme is the best among all coding schemes.}
    \label{fig:fig4}
\end{figure}

    At last, we have a discussion at the hardware efficiency of our coding scheme. We use computing
resource usage as the metric. Resource usage is caculated by average iteration:
$$resource\_usage = \frac{\sum_{i \in workers}computing\_time_{i}}{\sum_{i \in workers}total\_time_{i}}$$
As we can see, \textit{Naive} has a resource usage lower than $20\%$ in Fig.\ref{fig:fig4}. This is
incurred by low-computing capacity workers and many other factors, e.g., background
interferring process and fluctuate network. \textit{Cylic} coding scheme mitigates this
problem by discarding stragglers. However, it still has a limits incurred by unbalanced distribtuion
of computing resource. Our \textit{heter-aware} coding scheme and \textit{group-based} coding scheme
solve all these two problems and acheive high resource usage. Though still half of resouce is idle
due to communication overhead, this can be solved by combined techniques proposed by
\cite{zhang2017poseidon} that code gradients layer by layer.

\section{Conclusion}
    To tolerate stragglers and take fully advantage of computing resources, we propose two new
coding schemes in this paper, \textit{heter-aware} and \textit{group-based} coding scheme.
Traditional coding methods proposed by \cite{tandon2017gradient} could efficiently mitigate
stragglers, especially for fault tolerance, but their equivalent data allocation mechanism
causes that they have bad performance in heterogeneous clusters. Considering these, our
coding schemes take both stragglers and heterogeneity into account to tolerate stragglers
by firstly allocating data partitions to workers according to their processing speed
and then designing corresponding coding strategy. Evaluations show that our coding schemes
could acheive up to $3\times$ speedup compared to \textit{cyclic} coding scheme.

\bibliographystyle{ieeetr}
\bibliography{reference}
\end{document}